\documentclass[11 pt]{article}
\usepackage{amsmath, mathtools}
\usepackage{amsthm}
\usepackage{amssymb}
\usepackage{graphicx}
\usepackage{appendix}
\usepackage{color}
\usepackage{hyperref}

\usepackage{float}

\usepackage{multirow}

\newtheorem{theorem}{Theorem}[section]
\newtheorem{definition}[theorem]{Definition}

\newtheorem{lemma}[theorem]{Lemma}
\newtheorem{corollary}[theorem]{Corollary}
\newtheorem{remark}[theorem]{Remark}

\numberwithin{equation}{section}

\begin{document}

\title{Fractional Barndorff-Nielsen and Shephard model: applications in variance and volatility swaps, and hedging}
\author{ Nicholas Salmon\footnote{Email: nicholas.salmon@ndsu.edu}, Indranil SenGupta\footnote{Email: indranil.sengupta@ndsu.edu} \\ Department of Mathematics \\ North Dakota State University \\ Fargo, North Dakota, USA.}
\date{\today}
\maketitle

\begin{abstract}

In this paper, we introduce and analyze the fractional Barndorff-Nielsen and Shephard (BN-S) stochastic volatility model. The proposed model is based upon two desirable properties of the long-term variance process suggested by the empirical data: long-term memory and jumps. The proposed model incorporates the long-term memory and positive autocorrelation properties of fractional Brownian motion with $H>1/2$, and the jump properties of the BN-S model. We find arbitrage-free prices for variance and volatility swaps for this new model. Because fractional Brownian motion is still a Gaussian process, we derive some new expressions for the distributions of integrals of continuous Gaussian processes as we work towards an analytic expression for the prices of these swaps. The model is analyzed in connection to the quadratic hedging problem and some related analytical results are developed. The amount of derivatives required to minimize a quadratic hedging error is obtained. Finally, we provide some numerical analysis based on the VIX data. Numerical results show the efficiency of the proposed model compared to the Heston model and the classical BN-S model.

\end{abstract}



\section{Introduction}

The study of volatility has become a fundamental element in the finance literature since Chicago Board Option Exchange (CBOE) volatility index (VIX) was publicly announced in 1993. Empirical data in the financial analysis suggest that the volatility of asset prices is persistent. This means that the associated autocorrelation function decays very slowly (for example, see \cite{Baillie10, Baillie20, ContRama, Ding}). For this reason, it has become increasingly important to implement stochastic volatility models that enjoy the long memory property. Consequently, for such processes, the autocorrelation function is not integrable. In the pioneering work \cite{Comte}, the authors introduce a continuous-time model of stochastic volatility model with the
long memory property. The corresponding volatility process is called a ``trending stochastic volatility".  In their work they consider a fractional stochastic volatility model, which is based on a fractional Ornstein-Uhlenbeck process, driven by a fractional Brownian motion with Hurst parameter $H >\frac{1}{2}$.

It is well known (see, for example, \cite{anatobook}) that fractional Brownian motion (fBm) is neither a semimartingale nor a Markov process unless $H=\frac{1}{2}$. In that case, it coincides with the standard Brownian motion. The trajectories of fBm are H\"older continuous up to order $H$. Consequently, the quadratic variation of fBm is zero for $H>\frac{1}{2}$, and infinity for $H<\frac{1}{2}$. The Hurst index ($H$) describes the feature of the path of the fractional Brownian motion that is associated with it. A value of $H=\frac{1}{2}$ corresponds to non-correlating increments. A value $H \neq \frac{1}{2}$ corresponds correlating increments. In particular, $H>\frac{1}{2}$ corresponds to increments with positive correlation, and $H< \frac{1}{2}$ corresponds to increments with negative correlation.

In \cite{Comte2} a long memory extension of the Heston model is developed with the incorporation of a fractional Brownian motion. Long memory in the volatility process becomes much needed in order to justify the long-standing conundrums such as steep volatility smiles in long-term options and co-movements between implied and realized volatility. Volatility clustering, that can be observed through slowly decaying autocorrelations for absolute returns, is an important feature of most time series of financial asset returns. This can be argued to have been caused by the ``long-range dependence" of the stochastic volatility dynamics.

We note the contrast to the idea of ``rough” fractional volatility, which amounts to taking $H < \frac{1}{2}$.  A trending stochastic volatility model (with Hurst parameter $H > \frac{1}{2}$) is not efficient in describing the empirical short-time behavior of the volatility. In \cite{Gatheral} shows that for a short-time-interval, log-volatility behaves essentially as a fractional Brownian motion with Hurst exponent $H\sim 1/20$. This leads to the introduction of a rough stochastic volatility model with $ H<\frac{1}{2}$.  This has been studied in more detail in many works such as \cite{Alos, Bayer2, Bayer}. For the short-term, such models are very effective. In \cite{Bayer}, it is shown that the actual SPX variance swap curves are consistent with model forecasts, with particular examples from the weekend of the collapse of Lehman Brothers and the Flash Crash.

These two regimes for the values of $H$ (i.e., $H> \frac{1}{2}$ and $H<\frac{1}{2}$) can be blended, so that rough volatility governs the short time behavior while trending volatility affects the long time behavior. For example, in \cite{Bennedsen} the authors introduce a class of continuous-time stochastic volatility model of asset prices. The model can simultaneously incorporate roughness, for the short time, and slowly decaying autocorrelations, including proper long memory, which are two stylized facts often found in volatility data.

A proper statistical model of stochastic volatility is very useful in the analysis of variance and volatility swaps. Swaps written on the market volatility are financial instruments that are becoming increasingly useful for hedging and speculation of volatility. In recent literature (see \cite{sub1,sub2}) variance swaps are implemented to reduce the quadratic hedging errors of oil commodities.  In general, a swap is a financial derivative in which two counterparties exchange cash flows of two securities, interest rates, or other financial instruments for the mutual benefit of the exchangers. Volatility and variance swaps, respectively, are forward contracts on future realized volatility and the square of future realized volatility. Analysis of variance and volatility swaps are developed in connection to the underlying stochastic volatility models.  The literature devoted to volatility and variance swap developed rapidly during the past few years. For example, in \cite{anato1} a new probabilistic approach using the Heston model is proposed to study various swaps that depend on the volatility.  In \cite{anato3} variance swaps with delay are analyzed. In addition, analytical approximate formal asymptotic forms are obtained for expectation and variance of the realized continuously sampled variance for stochastic volatility with delay. In \cite{A4} the authors present a variance drift-adjusted version of the Heston model, that improves variance swap pricing based on the classical Heston model. In \cite{Kim} the authors choose the log-normal SABR model with fractional stochastic volatility to price variance and volatility swaps. A closed-form exact solution for the fair strike price is obtained in that work.

The literature described in the preceding paragraph do not incorporate the big fluctuations (the ``jumps") of a financial market. Many important features related to ``jumps" in an empirical market can be captured by models in which stochastic volatility of log-returns is constructed through Ornstein-Uhlenbeck (OU) type stationary stochastic process driven by a subordinator, where a subordinator is a L\'evy process with positive increments and no Gaussian component. Using L\'evy processes as driving noise, a large family of mean-reverting jump processes with linear dynamics can be constructed. The model that incorporates such processes offers the possibility of capturing important distributional deviations from Gaussianity and thus these are more practical models of dependence structures. This model is introduced in various works (see \cite{BN-S1, BN-S2}) of Barndorff-Nielsen and Shephard and is known in modern literature as the BN-S model. In \cite{vo1} the swaps written on powers of realized volatility for the BN-S model are analyzed. In \cite{ISGSH} the arbitrage-free pricing of variance and volatility swaps for BN-S model is studied.  In \cite{ISGSH} various approximate expressions for the pricing of volatility and variance swaps are obtained that provide a computational basis for the arbitrage-free variance and volatility swaps price. In addition, it is shown that for the approximate formulas, the error estimation in fitting the delivery price is much less than the existing models with comparable parameters. The analysis is further developed in \cite{AzizIndra}, where the bounds of the arbitrage-free variance swap price are found.

In this paper, we aim to combine the benefits of long-term memory of fractional Brownian motion with the jumps of the BN-S model. We improve the classical BN-S model by incorporating a fractional Ornstein-Uhlenbeck process, along with jumps driven by a subordinator, as is done in the classical BN-S model. We look to find arbitrage-free prices for variance and volatility swaps for this new model. Because fractional Brownian motion is still a Gaussian process, we derive some new expressions for the distributions of integrals of continuous Gaussian processes as we work towards an analytic expression for the prices of these swaps. The model is analyzed in connection to the quadratic hedging problem and some related analytical results are developed. We make use of these analytic expressions in simplified numerical computations to compare the performance of this new fractional model with previous popular models such as the classical BN-S model and the Heston model.

The paper is organized as follows.  In Section \ref{sec2} we discuss the mathematical and financial preliminaries that are useful for developing the analysis of this paper. In Section \ref{sec3} we propose the mathematical model, \emph{fractional  Barndorff-Nielsen and Shephard (BN-S) stochastic volatility model}, and implement that to analyze the variance and volatility swaps. In Section \ref{sec4} we derive some useful expressions for the conditional distribution \emph{fractional BN-S stochastic volatility model}. In addition, in Section \ref{sec4} we implement the proposed model to find an optimal hedging strategy. In Section \ref{sec5} various numerical results are provided. Finally, a brief conclusion is provided in Section \ref{sec6}.

\section{Some preliminaries with a basic model}
\label{sec2}

In this section, we introduce the mathematical preliminaries for this paper. In addition, we analyze a simplistic fractional log variance model. We conclude this section with some discussions on the variance and volatility swap.

\subsection{Mathematical Preliminaries} 
\label{sec21}

A random variable $X:\Omega\rightarrow\mathbb{R}^n$ is called a multivariate Gaussian, or just Gaussian, with mean $\mu$ and covariance matrix $\Sigma$, if it has a normal probability density function with mean $\mu$ and covariance matrix $\Sigma$. A Gaussian process is a stochastic process where all finite-dimensional distributions are Gaussian. There are two functions which characterize Gaussian processes- $\mu(t):=E[X_t]$, and $\rho(s,t) :=E[\langle X_s-\mu(s),X_t-\mu(t)\rangle]$.  These are called the mean and covariance functions respectively. If $\mu(t)=0$, $X_t$ is said to be centered. For the remainder of this article, we will deal with processes in $\mathbb{R}$, so the inner product will simply be a product. A well-known result affirms the following (see, for example, \cite{Kallenberg}).

\begin{theorem}
    Two Gaussian processes have the same law if and only if they have the same mean and covariance functions.
\end{theorem}

Another well-known theorem provides the foundation of the process $W^H_t$ that will be considered throughout this paper.
\begin{theorem}
    Let $0<H\leq 1$. Then there exists a continuous centered Gaussian process $W_t^H$ with covariance function given by
    \[\rho^H(s,t):=\frac{1}{2}(|t|^{2H}+|s|^{2H}-|t-s|^{2H}).\]
    Its sample paths are $\alpha$ H\"older for all $0<\alpha<H$.
\end{theorem}
\begin{proof}
A proof can be found in \cite{Nourdin}.
\end{proof}

We call this process the \emph{Fractional Brownian Motion}. For $H=1/2$ this is exactly the standard Brownian Motion. For the rest of the paper, if $H=1/2$, we denote $W_t^{1/2}$ by $W_t$. 

We recall that It\^o integration is necessary to make sense of differential equations of the form $dX_t=f(X_t,t)\,dW_t+g(X_t,t)\,dt$, because in general integrals of the form $\int_{0}^{t}f(X_s,s)dW_s$ are not well defined Riemann integrals. This arises from $W_t$ having infinite variation. However, it was shown in the pioneering work \cite{Young}, that if $\alpha+\beta>1$, the so called Young integral
\[\int_{a}^{b}f(x)dg(x):=\lim_{|\mathcal{P}|\rightarrow 0}\sum_{[u,v]\in\mathcal{P}}f(u)(g(v)-g(u)),\]
exists if $f$ and $g$ are $\alpha$ and $\beta$ H\"older respectively. Here we are using the notation that $\mathcal{P}$ is a partition of $[a,b]$ and $|\mathcal{P}|=\max_{[u,v]\in\mathcal{P}}v-u$. Unfortunately, for $H\leq 1/2$, this means that almost surely the integral $\int_{a}^{b}W_t^HdW_t^H$ does not exist in this Young sense. For $H=1/2$, It\^o integration solves this problem by taking advantage of properties of martingales. In this article we assume $H>1/2$ and are thus able to use this notion of the Young integral to define our differential equations.

\begin{definition}
Let $g(t)$ be $\alpha>1/2$ H\"older continuous. Let $\sigma(t)$ be $C^2$ with bounded derivatives and $\mu(t)$ Lipschitz. We say $x(t)$ solves the differential equation
\[dx(t)=\sigma(x(t))dg(t)+\mu(x(t))dt\]
in the Young sense if
\[x(t)-x(0)=\int_{0}^{t}\sigma(x(u))dg(u)+\int_{0}^{t}\mu(x(u))du,\]
where we assume a valid solution $x(t)$ is $\beta>1-\alpha$ H\"older so that the integrals are well defined Young integrals.
\end{definition}

\begin{theorem}
    Given an initial condition there exist unique solutions to Young differential equations as described above.
\end{theorem}
\begin{proof}
A proof can be found in \cite{Nourdin}.
\end{proof}

Next, we extend this notion to stochastic processes with almost surely $\alpha>1/2$ sample paths. 
\begin{definition}
Suppose that $G_t$ is almost surely $\alpha>1/2$ H\"older. Then we say that $X_t$ solves the differential equation
\begin{equation*}
dX_t=\sigma(X_t)dG_t+\mu(X_t)dt,
\end{equation*}
in the Young sense if for almost all $\omega\in\Omega$ we have that $X_t(\omega)$ solves the differential equation in the deterministic Young sense.
\end{definition}

\subsection{Fractional log variance model}
\label{sec22}

With the preliminaries in the last subsection, we introduce a fractional log variance model. Let $H>1/2$ and consider the corresponding Fractional Brownian Motion with $W_0^H=0$ almost surely. We note that $W_t^H$ is $\alpha$-H\"older for all $\alpha\in(0,H)$, and since $H>1/2$, we can interpret the following model in a Young sense.
\begin{equation*}
\sigma_t^2 =e^{X_t}, \quad       dX_t=\alpha(m-X_t)dt+\nu dW_t^H.
\end{equation*}

Clearly, $\nu$ is $C^2$ with bounded derivatives and $m-X_t$ is Lipshitz in $X_t$, hence we know unique solutions exist to the above differential equation almost surely.

\begin{lemma}
    The process 
    \[X_t=e^{-\alpha t}(X_0-m)+m+\nu W_t^H-\alpha\int_{0}^{t}e^{-\alpha(t-s)}\nu W_s^Hds,\] 
    solves $ dX_t=\alpha(m-X_t)dt+\nu dW_t^H$, in the Young sense.
\end{lemma}
\begin{proof}
    First recall that if $f(t)$ is differentiable, then $d(f(t))=f'(t)dt$. Then it follows easily that: 
    \begin{align*}
        dX_t&=d\left(e^{-\alpha t}(X_0-m)+m+\nu W_t^H-\alpha\int_{0}^{t}e^{-\alpha(t-s)}\nu W_s^Hds\right)dt\\
        &=\alpha\left(m-\left(e^{-\alpha t}(X_0-m)+m+\nu W_t^H-\alpha\int_{0}^{t}e^{-\alpha(t-s)}\nu W_s^H ds\right)\right)dt+\nu dW_t^H\\
        &=\alpha(m-X_t)dt+\nu dW_t^H.
    \end{align*}
\end{proof}

\begin{lemma}
\label{lemma27}
    Let $X_t\in\mathbb{R}$ be a continuous Gaussian process with covariance function $\rho(s,t)$ and $\mu(t) =E[X_t]$. Then $Y_t=\int_{0}^{t}X_sds$ is Gaussian with covariance $\int_{0}^{t}\int_{0}^{s}\rho(x,y)dxdy$ and mean $\int_{0}^{t}\mu(x)dx$.
\end{lemma}
\begin{proof}
    For each $N$, consider the multivariate Gaussian $X_t^N:=(X_{t_1},\ldots,X_{t_N})$ with mean $\mu_N=(E(X_{t_1}),\ldots, E(X_{t_N})),$ and covariance $\Sigma_N$, where $t_i=\frac{t}{N}i$. A standard theorem of Gaussian random variables says that if $X$ is multivariate Gaussian and $B$ is a linear transformation, then $BX$ is multivariate with mean $B\mu$ and covariance $B\Sigma B^T$. We consider the linear transformation $B_t^N$ defined as $B_t^N(x_1,\cdots,x_N):=\frac{t}{N}\sum_{i=1}^{N}x_i$.
    Thus, $B_t^NX_t^N$ is Gaussian with mean $\frac{t}{N}\sum_{i=1}^{N}\mu(t_i)$. The covariance of $B_t^NX_t^N$ and $B_s^NX_s^N$ is:

    \begin{align*}
  &      E\left[\left(B_s^NX_s^N-\frac{s}{N}\sum_{i=1}^{N}\mu(s_i)\right)\left(B_t^NX_t^N-\frac{t}{N}\sum_{i=1}^{N}\mu(t_i)\right)\right] \\
&=E\left[\left(\frac{s}{N}\sum_{i=1}^{N}X_{s_i}-\mu(t_i)\right)\left(\frac{t}{N}\sum_{i=1}^{N}X_{t_i}-\mu(t_i)\right)\right]\\
        &=\frac{st}{N^2}\sum_{i,j=1}^{N}E[(X_{s_i}-\mu(s_i))(X_{t_j}-\mu(t_j))]=\frac{st}{N^2}\sum_{i,j=1}^{N}\rho(s_i,t_j).
    \end{align*}
Since $X_t$ is almost surely continuous, for almost all $\omega$ we have $B_t^NX_t^N(\omega)\rightarrow\int_{0}^{t}X_s(\omega)ds$ as $N\rightarrow\infty$. Since the mean converges to $\int_{0}^{t}\mu(x)dx$ and the covariance function converges to \\
$\int_{0}^{t}\int_{0}^{s}\rho(x,y)dxdy$, $\int_{0}^{t}X_sds$ is Gaussian as well, with the corresponding mean and covariance function.
\end{proof}

It is immediate from Lemma \ref{lemma27} that $-\alpha\int_{0}^{t}e^{-\alpha (t-s)}\nu W_s^H ds$ is Gaussian with mean 0 and covariance 
\[(\alpha\nu)^2\int_{0}^{t}\int_{0}^{s}e^{-\alpha((s-x)+(t-y))}\rho^H(x,y)dxdy.\]
We also have that $e^{-\alpha t}(X_0-m)+m+\nu W_t^H$ is Gaussian with mean $e^{-\alpha t}(X_0-m)+m$, and covariance $\rho^H(s,t)$. If we are able to compute the covariance between $e^{-\alpha t}(X_0-m)+m+\nu W_t^H$ and $-\alpha\int_{0}^{t}e^{-\alpha (t-s)}\nu W_s^H ds$, then we will have a complete description of the Gaussian process $X_t$. For this we need the following lemma.
\color{black}

\begin{lemma}
    Let $X_t\in\mathbb{R}$ be a continuous Gaussian process with mean $\mu(t)$ covariance function $\rho(s,t)$. Then $E\left[X_s\int_{0}^{t}X_rdr\right]=\int_{0}^{t}\rho(s,x)-\mu(s)\mu(x)dx$.
\end{lemma}
\begin{proof}
    Let
    \[B_t^N(x_1,\ldots,x_N):=\frac{t}{N}\sum_{i=1}^{N}x_i.\]
    We now examine $X_sB_t^NX_t^N$ where $X_t^N:=(X_{t_1},\ldots,X_{t_N})$ and $t_i=\frac{t}{N}i$. From the proof of the previous lemma we know $X_sB_t^NX_t^N\rightarrow X_s\int_{0}^{t}X_rdr$ almost surely. We also have
    \begin{align*}
        E\left[X_s\frac{t}{N}\sum_{i=1}^{N}X_{t_i}\right]=\frac{t}{N}\sum_{i=1}^{N}E[X_sX_{t_i}] =\frac{t}{N}\sum_{i=1}^{N}\rho(s,t_i)-\mu(s)\mu(t_i) \rightarrow\int_{0}^{t}\rho(s,x)-\mu(s)\mu(x)dx.
    \end{align*}
    By continuity we get the result.
\end{proof}

\begin{theorem}
\label{theorem29}
 The processs   $X_t=e^{-\alpha t}(X_0-m)+m+\nu W_t^H-\alpha \int_{0}^{t}e^{-\alpha(t-s)}\nu W_s^H ds$, is a Gaussian process with mean $e^{-\alpha t}(X_0-m)+m$ and covariance function
\begin{align*}
\rho(s,t) & =\nu^2\left(\rho^H(s,t)-\alpha\left( \int_{0}^{t}e^{-\alpha(t-x)}\rho^H(s,x)dx+\int_{0}^{s}e^{-\alpha(s-x)}\rho^H(t,x)dx\right) \right. \\
& \left. +  \alpha^2\int_{0}^{t}\int_{0}^{s}e^{-\alpha((s-x)+(t-y))}\rho^H(x,y)dxdy\right).
 \end{align*}
\end{theorem}
\begin{proof}
    First we note that $e^{\alpha t}\nu W_t^H$ is a Gaussian process with mean 0 and covariance $\rho(s,t)=e^{\alpha(s+t)}\nu^2\rho^H(s,t)$.
    We have
    \begin{align*}
        &\text{Cov}\left(e^{-\alpha s}(X_0-m)+m+\nu W_s^H,-\alpha\int_{0}^{t}e^{-\alpha(t-r)}\nu W_r^H dr\right)\\
        &=E\left[-\alpha e^{-\alpha t}\nu W_s^H\int_{0}^{t}e^{\alpha r}\nu W_r^H dr\right]\\
        &=-\alpha e^{-\alpha(t+s)}E\left[e^{\alpha s}\nu W_s^H\int_{0}^{t}e^{\alpha r}\nu W_r^H dr\right]\\
        &=-\alpha e^{-\alpha(t+s)}\int_{0}^{t}e^{\alpha(s+x)}\nu^2\rho^H(s,x)dx =-\alpha\int_{0}^{t}e^{-\alpha(t-x)}\nu^2\rho^H(s,x)dx.
    \end{align*}
    This shows that the sum of the two Gaussian random variables $Y_t:=(e^{-\alpha t}(X_0-m)+m+\nu W_t^H)-\alpha\int_{0}^{t}e^{-\alpha(t-s)}\nu W_s^H ds$,
    is Gaussian with mean $e^{-\alpha t}(X_0-m)+m$ and that
    \begin{align*}
        &\text{Cov}(Y_s,Y_t)\\
        &=\text{Cov}(e^{-\alpha s}(X_0-m)+m+\nu W_s^H,e^{-\alpha t}(X_0-m)+m+\nu W_t^H)\\
        &\quad+\text{Cov}\left(e^{-\alpha s}(X_0-m)+m+\nu W_s^H,-\alpha\int_{0}^{t}e^{-\alpha(t-r)}\nu W_r^H dr\right)\\
        &\quad+\text{Cov}\left(-\alpha\int_{0}^{s}e^{-\alpha(s-r)}\nu W_r^H dr,e^{-\alpha t}(X_0-m)+m+\nu W_t^H\right) \\
&\quad+\text{Cov}\left(-\alpha\int_{0}^{s}e^{-\alpha(s-r)}\nu W_r^H dr,-\alpha\int_{0}^{t}e^{-\alpha(t-r)}\nu W_r^H dr\right)\\
        &=\nu^2\rho^{H}(s,t)-\alpha\int_{0}^{t}e^{-\alpha(t-x)}\nu^2\rho^H(s,x)dx\\
        &\quad-\alpha \int_{0}^{s}e^{-\alpha(s-x)}\nu^2\rho^H(t,x)dx+ (\alpha\nu)^2\int_{0}^{t}\int_{0}^{s}e^{-\alpha((s-x)+(t-x))}\rho^H(x,y)dxdy.
    \end{align*}
    This proves the result.
\end{proof}

Using Theorem \ref{theorem29}, we find that $X_t$ is normally distributed with mean $e^{-\alpha t}(X_0-m)+m$ and variance $\rho(t,t)$. Consequently, $E[e^{X_t}]=e^{e^{-\alpha t}(X_0-m)+m+\frac{1}{2}\rho(t,t)}$. This implies the following theorem: 
\begin{theorem}
\label{t1}
    Let $W_t^H$ be centered fractional Brownian motion with $H>1/2$, and consider the following fractional Ornstein-Uhlenbeck process:
    \[\quad dX_t=\alpha(m-X_t)dt+\nu dW_t^H.\]
    Then
\begin{equation}
\label{volWH}
E\left[\frac{1}{T}\int_{0}^{T}e^{X_t}dt\right]=\frac{1}{T}\int_{0}^{T}e^{e^{-\alpha t}(X_0-m)+m+\frac{1}{2}v(t)}dt,
\end{equation}
    where
\begin{equation}
\label{vt}
v(t)=\nu^2\left(\rho^H(t,t)-2\alpha \int_{0}^{t}e^{-\alpha(t-x)}\rho^H(t,x)dx+\alpha^2\int_{0}^{t}\int_{0}^{t}e^{-\alpha((t-x)+(t-y))}\rho^H(x,y)dxdy\right).
\end{equation}
\end{theorem}
\begin{proof}
    Since $X_t$ is Gaussian with mean $e^{-\alpha t}(X_0-m)+m$ and covariance $\rho(s,t)$, we have $E[e^{X_t}]=e^{e^{-\alpha t}(X_0-m)+m+\frac{1}{2}v(t)}$. Consequently,
    \begin{align*}
        E\left[\frac{1}{T}\int_{0}^{T}e^{X_t}dt\right]=\frac{1}{T}\int_{0}^{T}E\left[e^{X_t}\right]dt=\frac{1}{T}\int_{0}^{T}e^{e^{-\alpha t}(X_0-m)+m+\frac{1}{2}v(t)}dt.
    \end{align*}
\end{proof}

\subsection{Variance and volatility swaps}
\label{sec23}

In this paper, one of the primary goals is to analyze the variance and volatility swaps with respect to a new stochastic volatility model. For that reason, in this subsection, we provide preliminaries on the variance and volatility swaps. We denote the realized volatility by $\sigma_R(S)$. The subscript $R$ denotes the observed or realized volatility for some given underlying asset $S$. When the underlying asset is clear from the context, realized volatility is denoted simply as $\sigma_R$. 
If $\sigma_t$, $0 \leq t \leq T$ is the stochastic volatility for a given underlying asset $S$, then the realized volatility $\sigma_R$ over the life time of a contract is given by 
\begin{equation}
\label{1a}
\sigma_R=\sqrt{\frac{1}{T}\int_0^T\sigma_t^2dt}.
\end{equation}
Usually $\sigma_R$ is quoted in annual terms. The realized variance is $\sigma_R^2$ over the life of the contract is defined as 
\begin{equation}
\label{11}
\sigma_R^2=\frac{1}{T}\int_0^T\sigma_t^2\,dt.
\end{equation}
The following definitions can be found in \cite{ISGSH, Semere2}. 
\begin{definition}
A volatility swap is a  forward contract on the future realized volatility of a given underlying asset. The payoff of the volatility swap at the maturity $T$ is given by $N (\sigma_R -K_\text{Vol} )$,
where $K_\text{Vol}$  is the annualized volatility delivery price or exercise price, and $N$ is the notional amount of the swap in dollar per annualized volatility point.
\end{definition}
\begin{definition}
A variance swap is  a forward contract on realized variance, the square of the future realized volatility.  The payoff of variance swap at the maturity $T$ is given by 
$N(\sigma_R^2-\text{K}_{\text{Var}})$,
where $K_\text{Var}$ is the annualized delivery price or exercise price of the variance swap, and $N$ is the notional amount of the dollars per annualized volatility point squared. 
\end{definition}
Clearly, with respect to a risk-neutral measure $\mathbb{Q}$, the arbitrage-free price of the volatility swap is given by $P_\text{Vol}= E^{\mathbb{Q}}\left[e^{-rT}(\sigma_R-K_\text{Vol})\right]$, where $r$ is the risk-free interest rate. Similarly, $P_\text{Var}= E^{\mathbb{Q}}\left[e^{-rT}(\sigma_R^2-K_\text{Var})\right]$ provides the arbitrage free price of the variance swap. As observed in \cite{ISGSH, Semere2}, it turns out that  the calculation of $E^{\mathbb{Q}}(\sigma_R)$ is more involved than that of $E^{\mathbb{Q}}(\sigma_R^2)$. 
A useful estimate approximation regarding the expected value of realized volatility is obtained in \cite{anatobook} and is given by 
\begin{equation}
\label{fargo}
E^{\mathbb{Q}}[\sigma_R]= \displaystyle{E^{\mathbb{Q}}[\sqrt{\sigma_R^2}]\approx \sqrt{E^{\mathbb{Q}}[\sigma_R^2]}-\frac{\text{Var}^{\mathbb{Q}}(\sigma_R^2)}{8(E^{\mathbb{Q}}[\sigma_R^2])^{3/2}}}.
\end{equation}
However, this is \emph{not} a good estimate for the situation where the underlying asset undergoes big fluctuations. In the following sections, we derive an analytic expression for $E^{\mathbb{Q}}[\sigma_R^2]$ that can in general be implemented.

In the following, we summarize the calculations for variance swap for two well-known models- the Heston model, and the Barndorff-Nielsen and Shephard model.  For the Heston model the stock and volatility dynamics, with respect to a risk-neutral probability measure $\mathbb{Q}$, are given by 
\begin{equation*}
\frac{dS_t}{S_t}= r\,dt+ \sigma_t dW_t^1,
\end{equation*}
\begin{equation*}
d\sigma^2_t= k(\theta^2- \sigma^2_t)\,dt + \gamma \sigma_t\, dW_t^2,
\end{equation*}
where $r$ is the deterministic interest rate, $\sigma_0$ and $\theta$ are short and long volatility, $k>0$ is a reversion speed, $\gamma>0$ is a parameter, and $W_t^1$ and $W_t^2$ are independent standard Brownian motions. In this case the expected value and variance of $\sigma_R^2$ can be computed as
\begin{equation}
\label{hestonswap1}
E^\mathbb{Q}(\sigma_R^2)=\frac{1-e^{-kT}}{kT}(\sigma_0^2-\theta^2)+\theta^2,
\end{equation}
and
\begin{align}
\label{hestonswap2}
\text{Var}^\mathbb{Q}(\sigma_R^2) = \frac{\gamma^2e^{-2kT}}{2k^3T^2}[(2e^{2kT}-4e^{kT}kT-2)(\sigma_0^2-\theta^2)+(2e^{2kT}kT-3e^{2kT}+4e^{kT}-1)\theta^2].
\end{align}
Using \eqref{hestonswap1} in \eqref{11}, the value $E^{\mathbb{Q}}(\sigma_R^2)$ can be computed. Similarly, using \eqref{hestonswap1} and \eqref{hestonswap2} in \eqref{fargo}, it is possible to find the approximate value for $E^{\mathbb{Q}}(\sigma_R)$.

On the other hand, for the Barndorff-Nielsen and Shephard  model, it is considered that a riskless asset with constant return rate $r$ and a stock are traded up to a fixed horizon date $T$. It is assumed that the stock price ($S_t$) is defined on some filtered probability space $(\Omega, \mathcal{F}, (\mathcal{F}_t)_{0 \leq t \leq T}, \mathbb{P})$ and is given by:
\begin{equation}
\label{semere0}
S_t= S_0 \exp (X_t),
\end{equation}
\begin{equation}
\label{semere1}
dX_t = (\mu + \beta \sigma_t ^2 )\,dt + \sigma_t\, dW_t + \rho \,dZ_{\lambda t}, 
\end{equation}
\begin{equation}
\label{semere2}
d\sigma_t ^2 = -\lambda \sigma_t^2 \,dt + dZ_{\lambda t}, \quad \sigma_0^2 >0,
\end{equation}
where the parameters $\mu, \beta, \rho, \lambda \in \mathbb{R}$ with $\lambda >0$ and $\rho \leq 0$. $W= (W_t)$ is a Brownian motion and the process $Z= (Z_{\lambda t})$ is a subordinator. Also $W$ and $Z$ are assumed to be independent and $(\mathcal{F}_t)$ is assumed to be the usual augmentation of the filtration generated by the pair $(W, Z)$. 
It is shown in \cite{NV} that there exists a class of equivalent martingale measure ($\mathbb{Q}$) under which the structure of the BN-S model is preserved. A similar structure-preserving class of equivalent martingale measure in a general setting is also discussed in \cite{ijtaf}. With respect to this measure $\mathbb{Q}$, the following theorem is proved in \cite{ISGSH}.
\begin{theorem}
\label{arbitr1}
The arbitrage free price of the variance swap, with respect to a risk-neutral measure $\mathbb{Q}$, is given by
$\text{P}_{\text{Var}}=e^{-rT}\left[\frac{1}{T}\left(\lambda^{-1}\left(1-e^{-\lambda T}\right)\left(\sigma_0^2-\kappa_1\right)+\kappa_1 T\right)+ \rho^2 \lambda \kappa_2-\text{K}_{\text{Var}}\right]$,
where $\kappa_1$ and $\kappa_2$ are the first cumulant (i.e., the expected value) and  the second cumulant (i.e., the variance) of $Z_1$ respectively.
\end{theorem}

\section{The model and pricing of variance and volatility swaps}
\label{sec3}

In this section, we develop the fractional BN-S stochastic volatility model. We analyze variance and volatility swaps with respect to this model. 

\subsection{Fractional BN-S stochastic volatility model}
\label{sec31}

Here, we propose a stochastic volatility model that will be analyzed and implemented for the rest of this paper. We consider a frictionless financial market where a riskless asset with constant return rate $r$ and a stock are traded up to a fixed horizon date $T$. We assume that the derivative price $S= (S_t)_{t \geq 0}$, under some risk-neutral measure $\mathbb{Q}$ on some filtered probability space $(\Omega, \mathcal{F}, (\mathcal{F}_t)_{0 \leq t \leq T}, \mathbb{Q})$ is given by:

\begin{equation}
\label{se500}
S_t= S_0 \exp (Y_t),
\end{equation}
\begin{equation}
\label{se501}
dY_t = b_t \,dt + \sigma_t \, dW_t + \rho \,dZ_{ t}, \quad b_t= (c_1-c_2 \sigma_t^2),
\end{equation}
\begin{equation}
\label{se502}
  d\sigma_t^2= \lambda (e^{X_t}-\sigma_t^2)\,dt+ e^{X_t}\,dX_t+ dZ_t, 
\end{equation}
where 
\begin{equation}
\label{se503}
dX_t=\alpha(m-X_t)dt+\nu dW_t^H, \quad Z_0=0,
\end{equation}
and the parameters $c_1, c_2,, \rho  \in \mathbb{R}$ with $\rho \leq 0$ and $c>0$. In the above expressions $W_t$ is a standard Brownian motion, $W_t^H$ is a fractional Brownian motion with $H > \frac{1}{2}$, and the process $Z_t$ is a subordinator. In addition, $W_t$,  $W_t^H$, and $Z_t$ are assumed to be mutually independent processes. We call \eqref{se500}, \eqref{se501}, \eqref{se502}, and \eqref{se503} the \emph{fractional Barndorff-Nielsen and Shephard  (BN-S) stochastic volatility model}. Since $H>1/2$ is assumed, this model is focused on long-range dependence. Thus, this is a trending stochastic volatility model. From \eqref{se502} we observe that
\begin{align*}
e^{\lambda t} (d\sigma_t^2 + \lambda \sigma_t^2\,dt) & =  \lambda e^{\lambda t+X_t}\,dt+ e^{\lambda t +X_t}\,dX_t+ e^{\lambda t}dZ_t \\
d(e^{\lambda t} \sigma_t^2) &= d(e^{\lambda t +X_t})+ e^{\lambda t}dZ_t.
\end{align*}
Therefore
\begin{equation}
\label{sigmat}
\sigma_t^2= m_t + e^{X_t} + \int_0^t e^{-\lambda (t-s)}\,dZ_s, \quad \text{where} \quad m_t=e^{-\lambda t}(\sigma_0^2-e^{X_0}).
\end{equation}

\begin{theorem}
\label{t2}
For the \emph{fractional BN-S stochastic volatility model} 
\begin{equation}
\label{n2}
E^{\mathbb{Q}}(\sigma_R^2)= \frac{v_1(T)}{T}+ \frac{1}{T}\int_{0}^{T}e^{e^{-\alpha t}(X_0-m)+m+\frac{1}{2}v(t)}\,dt,
\end{equation}
where $v(t)$ is given by \eqref{vt}, and
\begin{equation}
\label{v1}
v_1(T)= \frac{1}{\lambda}\left((1-e^{-\lambda T})(\sigma_0^2-e^{X_0})+ \left(T-\lambda^{-1}(1-e^{-\lambda T})\right)\kappa_1\right),
\end{equation}
where  $\kappa_1$ is the first cumulant of $Z_1$.
\end{theorem}
\begin{proof}
From \eqref{sigmat}, we can obtain
\begin{equation}
\label{se44}
\sigma_R^2=\frac{1}{T} \left(\lambda^{-1}(1-e^{-\lambda T})(\sigma_0^2-e^{X_0}) +\lambda^{-1}\int_0^T\left(1-e^{-\lambda(T-s)}\right)dZ_{\lambda s} \right)+ \frac{1}{T}\int_0^T e^{X_t}\,dt.
\end{equation}
Consequently,
\begin{align*}
& E^{\mathbb{Q}}(\sigma_R^2)= \frac{1}{T} \left(\lambda^{-1}(1-e^{-\lambda T})(\sigma_0^2-e^{X_0})+\lambda^{-1}\kappa_1  \int_0^T\left(1-e^{-\lambda(T-s)}\right)\, ds\right) + \frac{1}{T}\int_0^T E^{\mathbb{Q}} (e^{X_t})\,dt \\
&=\frac{1}{\lambda T}\left((1-e^{-\lambda T})(\sigma_0^2-e^{X_0})+ \left(T-\lambda^{-1}(1-e^{-\lambda T})\right)\kappa_1\right)+  \frac{1}{T}\int_{0}^{T}e^{e^{-\alpha t}(X_0-m)+m+\frac{1}{2}v(t)}\,dt,
\end{align*}
where  $\kappa_1$ is the first cumulant (i.e., the expected value) of $Z_1$, and $v(t)$ is given by \eqref{vt} (in Theorem \ref{t1}). 
\end{proof}
The following result that follows immediately provides an arbitrage-free variance swap pricing formula for the \emph{fractional BN-S stochastic volatility model}.
\begin{corollary}
The arbitrage free price of the variance swap for the  the \emph{rough BN-S stochastic volatility model} is given by
\begin{equation}
\label{var}
\text{P}_{\text{Var}}=\frac{e^{-rT}}{T}\int_{0}^{T}e^{e^{-\alpha t}(X_0-m)+m+\frac{1}{2}v(t)}\,dt + \mu_1(T),
\end{equation}
where $v(t)$ is given by \eqref{vt} and $\mu_1(T)=e^{-rT}\left(\frac{v_1(T)}{T}-\text{K}_{\text{Var}}\right)$, with $v_1(T)$ given by \eqref{v1}. 
\end{corollary}

\subsection{Pricing volatility swaps}
\label{sec32}

For simplicity, in this section we choose $\sigma_0^2=e^{X_0}$ for the \emph{fractional BN-S stochastic volatility model}. The intuition on this assumption is that we do not start the process shortly after a jump has already occurred. For a better understanding of this, consider what it means for $X_0\neq m$. This means that the fractional OU process has started away from its mean, and thus the term $e^{-\alpha t}(X_0-m)$ is non-zero. So we picture the expected process exponentially drifting back to its mean at rate $\alpha$ in this case. For the fractional BN-S model, looking at \eqref{se502}, we see that the drift term causes $\sigma_t^2$ to exponentially drift to $e^{X_t}$ at rate $\lambda$. What causes $\sigma_t^2$ to be different than $e^{X_t}$? This is exactly the jump term $dZ_t$. Hence we can view the appearance of $e^{-\lambda t}(\sigma_0^2-e^{X_0})$ in \eqref{sigmat} as the result of $\sigma_0^2$ being different than $e^{X_0}$, or that a jump as already occurred. So the assumption $\sigma_0^2=e^{X_0}$ is not restrictive, it is just the assumption that the initial condition starting away from the asymptotic mean $e^{m}$ is the result of a continuous drift of $e^{X_t}$ rather than a jump in $\sigma_t^2$. Thus for \eqref{sigmat}, we assume that $m_t=0$. Consequently,
$\sigma_t^2=  e^{X_t} + \int_0^t e^{-\lambda (t-s)}\,dZ_s$.

At this point, we introduce a couple of notations. For a stochastic process $A_t$, $0 \leq t \leq T$, we denote the characteristic function and the cumulant generating function of by $\Phi_{A_t}(\theta)= E(\exp(i\theta A_t))$ and $\kappa_{A_t}(\theta)= \log E(\exp(\theta A_t))$, respectively. The relation between the characteristic function and the cumulant generating function of $A_t$  is given by (see \cite{ISGSH}) $\Phi_{A_t}(\theta)= \exp[\kappa_{A_t}(i\theta)]$.
The moments of $A_t$ can be obtained from $\Phi_{A_t}(\theta)$ by $E^{\mathbb{Q}}(A_t^k)= (-i)^k \frac{d^k \Phi_{A_t}(\theta)}{d \theta^k} \Big|_{\theta=0}$, $k=1,2,\dots$.
We summarize the following lemma from \cite{ISGSH}.  
\begin{lemma}
\label{At}
Suppose that $A_t= \alpha+ \int_0^{\lambda t} (1-e^{-s})\, dZ_s$, where $\alpha \in \mathbb{R}$ is a constant, and $0 \leq t \leq T$. Then 
\begin{equation}
\label{phiAt}
\Phi_{A_t}(\theta)= \exp\left( i\theta \alpha+\int_0^{\lambda t} \kappa(i \theta (1-e^{-s}))\, ds\right),
\end{equation}
where $\kappa(\cdot)$ is the cumulant generating function for $Z_1$. The moments of $A_t$ are given by $E^{\mathbb{Q}}(A_t^k)= (-i)^k g_k(0)$, $k= 1,2,\dots$, where $g_1(\theta) = i \left(\alpha+\int_0^{\lambda t}   (1-e^{-s}) \kappa' (i \theta  (1-e^{-s}))\,ds\right)$, and $g_{k+1}(\theta)= g_1(\theta) g_k(\theta)+ g_{k}'(\theta)$, $k=1,2, \dots$.
\end{lemma}

We observe that $\kappa^{(k)}(0)= \kappa_k$, for $k=1,2,\dots$. Consequently, by using Lemma \ref{At}, it is possible to compute any moment for the process $A_t= \alpha+ \int_0^{\lambda t} (1-e^{-s})\, dZ_s$, in terms of cumulants of $Z$. For example, $E^{\mathbb{Q}}(A_t)= -i g_1(0)= \alpha+ \int_0^{\lambda t} \kappa_1(1-e^{-s})\,ds= \alpha+\kappa_1(\lambda t -1 + e^{-\lambda t})$. Similarly,
\begin{align*}
E^{\mathbb{Q}}(A_t^2) = (-i)^2 g_2(0) & = -(g_1(0)^2+ g_1'(0)) \\
&= \left(\alpha+\kappa_1(\lambda t -1 + e^{-\lambda t})\right)^2+ \kappa_2\left(2e^{-\lambda t}-\frac{3}{2}-\frac{1}{2}e^{-2\lambda t}+\lambda t\right).
\end{align*} 
Next, we observe that
\begin{align}
\label{eins}
    E^{\mathbb{Q}}\left[\left(\frac{1}{T}\int_{0}^{T}e^{X_t}dt\right)^n\right]&=E^{\mathbb{Q}}\left[\frac{1}{T}\int_{0}^{T}e^{X_{t_1}}dt_1\cdots\frac{1}{T}\int_{0}^{T}e^{X_{t_n}}dt_n\right] \nonumber \\
    &=E^{\mathbb{Q}}\left[\frac{1}{T^n}\int_{0}^{T}\cdots\int_{0}^{T}e^{X_{t_1}+\cdots+X_{t_n}}dt_1\cdots dt_n\right] \nonumber \\
    &=\frac{1}{T^n}\int_{[0,T]^n}E^{\mathbb{Q}}\left[e^{X_{t_1}+\cdots+X_{t_n}}\right]dt_1\cdots dt_n.
\end{align}
This again allows us to use the structure of multivariate Gaussian's to compute desired quantities related to continuous Gaussian processes. Since the sum of the $X_{t_i}$ is Gaussian and we know its mean and covariance functions, we have
\begin{equation*}
E^{\mathbb{Q}}\left[e^{\sum_{i=1}^{n}X_{t_i}}\right]=e^{\sum_{i=1}^{n}\mu(t_i)+\frac{1}{2}\sum_{i,j=1}^{n}\rho(t_i,t_j)}.
\end{equation*}
So using this in \eqref{eins} we obtain:
\begin{equation}
\label{goest}
E^{\mathbb{Q}}\left[\left(\frac{1}{T}\int_{0}^{T}e^{X_t}dt\right)^n\right]=\frac{1}{T^n}\int_{[0,T]^n}e^{\sum_{i=1}^{n}\mu(t_i)+\frac{1}{2}\sum_{i,j=1}^{n}\rho(t_i,t_j)}dt_1\cdots dt_n.
\end{equation}

We introduce $\beta>0$ such that $\sigma_R^2<2\beta^2$. This is possible for any set of empirical financial data. 
\begin{theorem}
\label{sigmaR4}
Assume that $\sigma_R^2<2\beta^2$, for some $\beta>0$. Also, assume that $\sigma_0^2= e^{X_0}$. Then the arbitrage free value of the volatility swap is given by
\begin{equation}
\label{newww}
\text{P}_{\text{Vol}}  = e^{-rT}\left(\beta\sum_{n=0}^{\infty}\frac{(-1)^{n+1}(2n)!}{4^n(n!)^2(2n-1)}\frac{1}{\beta^{2n}}\sum_{k=0}^{n}A_{n-k}B_k -\text{K}_{\text{Vol}}\right),
\end{equation}
where
\begin{equation}
\label{AAn}
    A_k = E^{\mathbb{Q}}\left[\left(-\beta^2 + \frac{1}{\lambda T} \int_0^{\lambda T}\left(1-e^{-s}\right)dZ_{s}\right)^k\right],
\end{equation}
and 
\begin{equation}
\label{BBn}
    B_k=E^{\mathbb{Q}}\left[\left(\frac{1}{T}\int_{0}^{T}e^{X_t}dt\right)^k \right]=\frac{1}{T^n}\int_{[0,T]^k}e^{\sum_{i=1}^{k}\mu(t_i)+\frac{1}{2}\sum_{i,j=1}^{k}\rho(t_i,t_j)}dt_1\cdots dt_k,
\end{equation}
for $k=1,2,\dots$. The expression \eqref{AAn} for $A_k$, $k=1,2,\dots$, can be computed by using Lemma \ref{At}.
\end{theorem}
\begin{proof}
With $m_t=0$, we have the following:
\begin{equation*}
\sigma_R^2=\frac{1}{\lambda T}\int_{0}^{T} \left(1-e^{-\lambda(T-s)}\right)dZ_s+\frac{1}{T}\int_{0}^{T}e^{X_t}dt.
\end{equation*}
If $\sigma_R^2<2\beta^2$, then we notice that $\sigma_R^2=\beta^2\left(1+\frac{\sigma_R^2-\beta^2}{\beta^2}\right)$, with $|\frac{\sigma_R^2-\beta^2}{\beta^2}|<1$.
We note that for $|x|<1$, 
\[\sqrt{x+1}=\sum_{n=0}^{\infty}\frac{(-1)^{n+1}(2n)!}{4^n(n!)^2(2n-1)}x^n.\]
Consequently, $E^{\mathbb{Q}}\left[\sqrt{\sigma_R^2}\right]=\beta E^{\mathbb{Q}}\left[\sqrt{1+\frac{\sigma_R^2-\beta^2}{\beta^2}}\right]$. 
Since we are assuming $\left|\frac{\sigma_R^2-\beta^2}{\beta^2}\right|<1$, we can use our ability to compute $E^{\mathbb{Q}}\left[\left(\frac{\sigma_R^2-\beta^2}{\beta^2}\right)^n\right]$ in order to evaluate the expectation of the realized volatility. We find
\[\sigma_R^2-\beta^2=-\beta^2+\frac{1}{\lambda T}\int_{0}^{T}\left(1-e^{-\lambda(T-s)}\right)dZ_s+\frac{1}{T}\int_{0}^{T}e^{X_t}dt.\]
We note that 
\begin{align}
\label{distt}
\frac{1}{\lambda T} \int_0^T\left(1-e^{-\lambda(T-s)}\right)dZ_{\lambda s}  \,{\buildrel d \over =}\, \frac{1}{\lambda T} \int_0^{\lambda T}\left(1-e^{-s}\right)dZ_{s},
\end{align}
where $\,{\buildrel d \over =}\,$ denotes the equality in distribution. 
We denote for $n=1,2,\dots$,
\begin{align*}
    A_n  =E^{\mathbb{Q}}\left[\left(-\beta^2+\frac{1}{\lambda T}\int_{0}^{T}\left(1-e^{-\lambda(T-s)}\right)dZ_s\right)^n\right] = E^{\mathbb{Q}}\left[\left(-\beta^2 + \frac{1}{\lambda T} \int_0^{\lambda T}\left(1-e^{-s}\right)dZ_{s}\right)^n\right],
\end{align*}
where the last equality follows from \eqref{distt}. 
We also denote $B_n=E^{\mathbb{Q}}\left[\left(\frac{1}{T}\int_{0}^{T}e^{X_t}dt\right)^n \right]$, $n=1,2,\dots$. Thus, by using \eqref{goest} we obtain \eqref{BBn}. Consequently, $E^{\mathbb{Q}}\left[\left(\frac{\sigma_R^2-\beta^2}{\beta^2}\right)^n\right]=\frac{1}{\beta^{2n}}\sum_{k=0}^{n}A_{n-k}B_{k}$. Thus we obtain:
\begin{align*}
    E^{\mathbb{Q}}\left[\sqrt{\sigma_R^2}\right]=\beta E^{\mathbb{Q}}\left[\sqrt{1+\frac{\sigma_R^2-\beta^2}{\beta^2}}\right]=\beta\sum_{n=0}^{\infty}\frac{(-1)^{n+1}(2n)!}{4^n(n!)^2(2n-1)}\frac{1}{\beta^{2n}}\sum_{k=0}^{n}A_{n-k}B_k.
\end{align*}
\end{proof}

Next we consider the infinite series $\beta\sum_{n=0}^{\infty}\frac{(-1)^{n+1}(2n)!}{4^n(n!)^2(2n-1)}\frac{1}{\beta^{2n}}\sum_{k=0}^{n}A_{n-k}B_k$ of Theorem \ref{sigmaR4}. We find a rate of convergence of this series for a special annular region. For this, we assume that $0<\beta^2<\sigma_R^2<2\beta^2$. 
\begin{theorem}
For the annular region $0<\beta^2<\sigma_R^2<2\beta^2$, the $(N-1)$-th partial sum of $$\beta\sum_{n=0}^{\infty}\frac{(-1)^{n+1}(2n)!}{4^n(n!)^2(2n-1)}\frac{1}{\beta^{2n}}\sum_{k=0}^{n}A_{n-k}B_k,$$  has absolute error less than $\beta \frac{1}{(2N-1)\sqrt{3N+1}}$, in approximating $E^{\mathbb{Q}}\left[\sqrt{\sigma_R^2}\right]$.
\end{theorem}
\begin{proof}
With the condition, $\beta^2<\sigma_R^2<2\beta^2$, the infinite series representation of $E(\sigma_R)$ is an alternating series. In addition, we observe that,
\begin{equation}
\label{buf}
E^{\mathbb{Q}}\left[\left(\frac{\sigma_R^2-\beta^2}{\beta^2}\right)^n\right]=\frac{1}{\beta^{2n}}\sum_{k=0}^{n}A_{n-k}B_{k} <1,  \quad \text{ for each} \quad n.
\end{equation}
 Therefore, 
\begin{align*}
& \left|E^{\mathbb{Q}}(\sigma_R)- \beta\sum_{n=0}^{N-1}\frac{(-1)^{n+1}(2n)!}{4^n(n!)^2(2n-1)}\frac{1}{\beta^{2n}}\sum_{k=0}^{n}A_{n-k}B_k, \right| \\
& \leq \frac{\beta}{4^N (2N-1)}{2N \choose N},
\end{align*}
where we use properties of an alternating series and \eqref{buf}. It is well known that  ${2n \choose n} \leq \frac{4^n}{\sqrt{3n+1}}$, for all $n \geq 1$.
Therefore we obtain
\begin{align*}
\left|E^{\mathbb{Q}}(\sigma_R)- \beta\sum_{n=0}^{N-1}\frac{(-1)^{n+1}(2n)!}{4^n(n!)^2(2n-1)}\frac{1}{\beta^{2n}}\sum_{k=0}^{n}A_{n-k}B_k, \right|< \beta \frac{1}{(2N-1)\sqrt{3N+1}}.
\end{align*}
\end{proof}

The quantity $\beta$ can be used as a ``control parameter" that improves the rate of convergence of the infinite series $\beta\sum_{n=0}^{\infty}\frac{(-1)^{n+1}(2n)!}{4^n(n!)^2(2n-1)}\frac{1}{\beta^{2n}}\sum_{k=0}^{n}A_{n-k}B_k$, that appears on the right hand side of \eqref{newww}. 

\section{Quadratic hedging with fractional BN-S stochastic volatility model}
\label{sec4}

A perfect hedge is a position undertaken by an investor that would eliminate the risk of an existing position, or a position that eliminates all market risk from a portfolio. However, in real markets perfect hedges do not exist and options are not redundant. Thus we are forced to reconsider hedging in the more realistic sense of approximating a target payoff with a trading strategy. 
There are various approaches (see \cite{cont}) of hedging in an incomplete market such as Merton's approach, superhedging, utility maximization, quadratic hedging, etc. In this section, we are implementing techniques to minimize quadratic hedging errors when the market dynamics are modeled by the \emph{fractional BN-S stochastic volatility model}.

\subsection{Conditional distribution of $X_t$}

In this section we develop some results related to the Conditional distribution of $X_t$. The results will be implemented in the following sections. In \cite{conditional}, the following notations are introduced. Using the same $\alpha$ that is used in the  \emph{fractional BN-S stochastic volatility model} (in \eqref{se503}), we define:
\begin{equation*}
\left(I_{T-}^{H-1/2}f\right)(\tau) =\frac{1}{\Gamma(H-1/2)}\int_{\tau}^{T}f(s)(s-\tau)^{H-3/2}ds,
\end{equation*}
\begin{equation*}
\langle f,g\rangle_{H-1/2,T} =C(H)\int_{0}^{T}s^{-(2H-1)}\left(I_{T-}^{H-1/2}(\cdot)^{H-1/2}f(\cdot)\right)(s)\left(I_{T-}^{H-1/2}(\cdot)^{H-1/2}g(\cdot)\right)(s)ds,
\end{equation*}
where $C(H)= \frac{\pi(H-1/2)2H}{\Gamma(2-H)\sin(\pi(H-1/2))}$,
\begin{equation*}
||f||_{H-1/2,T}^2=\langle f,f\rangle_{H-1/2,T},
\end{equation*}
\begin{equation*}
\Phi_c(s,t,v)=\frac{\sin{\pi(H-1/2)}}{\pi}v^{-(H-1/2)}(s-v)^{-(H-1/2)}\int_{s}^{t}\frac{z^{H-1/2}(z-s)^{H-1/2}}{z-v}c(z)dz,
\end{equation*}
and
\begin{equation*}
 D(s,t)=\frac{1}{\alpha}(1-e^{-\alpha(t-s)}).
\end{equation*}
The following theorem is proved in \cite{conditional}.
\begin{theorem}
    Consider the process $X_t$ given by $dX_t=\nu dW^H_t+\alpha(m-X_t)dt$. Let $c(\tau)=\nu D(t, \tau)$. Then $X_{\tau}$ conditional on $\mathcal{F}_{t}$ is normally distributed with mean
\begin{equation}
\label{Mttt}
M(\tau, t)=X_{t}e^{-\alpha(\tau-t)}+m(1-e^{-\alpha(\tau-t)})+\int_{0}^{t}\Phi_c(t,\tau,s)dW_s^H,
\end{equation}
    and variance
\begin{equation}
\label{Nttt}
 V(\tau, t)=||c(\cdot)1_{[t, \tau]}||^2_{H-1/2,T}-||\Phi_c^{H-1/2}(t, \tau,\cdot)1_{[0,t]}(\cdot)||^2_{H-1/2,T}.
\end{equation}
    Therefore, 
\begin{equation}
\label{nicksen}
 E^{\mathbb{Q}}\left[e^{X_{\tau}}|\mathcal{F}_t\right]=e^{M(\tau, t)+\frac{1}{2}V(\tau, t)}.
\end{equation}
\end{theorem}

We denote 
\begin{equation}
\label{tildeW}
\tilde{W}_{t}(\tau)=\sqrt{2H}\int_{t}^{\tau}(\tau-s)^{H-1/2}dW_s, \quad \text{for} \quad \tau>t.
\end{equation}
We are interested in computing $E^{\mathbb{Q}}\left[\left.e^{\eta\tilde{W}_t(\tau)-\alpha\int_{t}^{\tau}X_sds}\right|\mathcal{F}_t\right]$. 

\begin{theorem}
For  $\tau>t$, we define $\Lambda(\tau, t)= E^{\mathbb{Q}}\left[ \exp\left( \eta \tilde{W}_t(\tau)   -\alpha\int_t^{\tau} X_s\,ds  \right) | \mathcal{F}_t\right]$.
Then, 
\begin{equation}
\label{Lammbda}
\Lambda(t, \tau)= \exp\left(M(\tau, t)+\frac{1}{2}V(\tau, t)- X_t -\alpha m (\tau-t)- \nu C_H N_t(\tau) \right),
\end{equation}
where $M(\tau,t)$ and $N(\tau,t)$ are given by \eqref{Mttt} and \eqref{Nttt} respectively, and 
\begin{equation}
\label{Nnnt}
N_t(\tau)= \int_{-\infty}^t \left(\frac{1}{(\tau-s)^{\gamma}}- \frac{1}{(t-s)^{\gamma}}\right)\, dW_s,
\end{equation}
\end{theorem}

\begin{proof}
The Mandelbrot-Van Ness representation of fractional Brownian motion $W^H$ is given by (see \cite{Bayer})
$$W_t^H= C_H \left( \int_{-\infty}^t \frac{dW_s}{(t-s)^{\gamma}} - \int_{-\infty}^0 \frac{dW_s}{(-s)^{\gamma}} \right),$$
where $\gamma= \frac{1}{2}-H$, and $C_H= \sqrt{\frac{2H \Gamma(\frac{3}{2}-H)}{\Gamma(H+ \frac{1}{2}) \Gamma(2-2H)}}$. 
We observe
\begin{align*}
\log e^{X_{\tau}}- \log e^{X_t} & = X_{\tau}-X_t \\
&= \nu(W_{\tau}^H- W_t^H)- \alpha \int_t^{\tau} (X_s-m)\,ds \\
& = \nu C_H \left( \int_{-\infty}^{\tau} \frac{dW_s}{(\tau-s)^{\gamma}} -\int_{-\infty}^t \frac{dW_s}{(t-s)^{\gamma}} \right)- \alpha \int_t^{\tau} (X_s-m)\,ds \\
&= \nu C_H (M_t(\tau) +N_t(\tau)) -\alpha\int_t^{\tau} X_s\,ds + \alpha m (\tau-t),
\end{align*}
where 
\begin{equation}
\label{Mmmt}
M_t(\tau)= \int_t^{\tau}\frac{dW_s}{(\tau-s)^{\gamma}},
\end{equation}
is independent of $\mathcal{F}_t$, and $N_t(\tau)$ given by \eqref{Nnnt} is $\mathcal{F}_t$-measurable. Consequently,
\begin{equation}
\label{psit}
e^{X_{\tau}}= e^{X_t+\alpha m (\tau-t)} \exp\left( \nu C_H (M_t(\tau) +N_t(\tau)) -\alpha\int_t^{\tau} X_s\,ds  \right).
\end{equation}
With $\tilde{W}_t(\tau)$ defined in\eqref{tildeW}, and $\eta= \frac{\nu C_H}{\sqrt{2H}}$, we obtain $\nu C_H M_t(\tau)= \eta \tilde{W}_t(\tau)$. Hence
\begin{align*}
E^{\mathbb{Q}}(e^{X_{\tau}}| \mathcal{F}_t) & = e^{X_t +\alpha m (\tau-t)+ \nu C_H N_t(\tau)  }E^{\mathbb{Q}}\left[ \exp\left( \eta \tilde{W}_t(\tau)   -\alpha\int_t^{\tau} X_s\,ds  \right) | \mathcal{F}_t\right] \\
& =  e^{X_t +\alpha m (\tau-t)+ \nu C_H N_t(\tau)  } \Lambda(\tau, t),
\end{align*}
where $\Lambda(\tau, t)= E^{\mathbb{Q}}\left[ \exp\left( \eta \tilde{W}_t(\tau)   -\alpha\int_t^{\tau} X_s\,ds  \right) | \mathcal{F}_t\right]$. Comparing the last expresssion with \eqref{nicksen}, we obtain \eqref{Lammbda}. 
\end{proof}

\subsection{Quadratic optimal hedging strategy}

In this subsection, we analyze an optimal hedging strategy in terms of quadratic hedging in connection to the \emph{fractional BN-S stochastic volatility model}. Quadratic hedging is a hedging strategy that minimizes the hedging error in the mean square sense (see \cite{sub1, hum, sub2}). For the simplicity of notation, we denote the arbitrage-free variance swap price ($\text{P}_{\text{Var}}(t, \sigma_t^2)$) by $P(t, \sigma_t^2)$. Before proving the main result related to quadratic hedging, we provide a representation of the variance swap price in terms of two functions- one is dependent on the continuous dynamics of the stock price, and the other is dependent on the significant fluctuations (``jumps") of the stock price.

\begin{theorem}
\label{thm1}
Consider the \emph{fractional BN-S stochastic volatility model} given by \eqref{se500}, \eqref{se501}, \eqref{se502}, and \eqref{se503}. Then, the arbitrage free value of $P(t, \sigma_t^2)$, with respect to the equivalent martingale measure $\mathbb{Q}$, is almost surely given by
\begin{align}
\label{mainn}
P(t, \sigma_t^2)= P_1(t)+ P_2(t,Z_t),
\end{align}
where 
\begin{equation}
\label{p1}
P_1(t) = \frac{e^{-r(T-t)}}{T}\left(\int_t^T \exp(M(s, t)+\frac{1}{2}V(s, t)) \, ds + \int_0^t \exp(e^{-\alpha s}(X_0-m)+m+\frac{1}{2}v(s))\, ds \right), 
\end{equation}
and $P_2$ satisfies the integro-differential equation
\begin{equation}
\label{p2eq}
-rP_2+ \frac{\partial P_2}{\partial t} + \int_0^{\infty} \left(P_2(t, Z_{t-}+y)- P_2(t, Z_{t-})- y \frac{\partial P_2}{\partial z} 1_{|y| \leq 1}\right) \nu(dy)=0.
\end{equation} 
with the final condition
\begin{equation}
\label{finalcon}
 P_2(T,  Z_T)= E^{\mathbb{Q}}(\sigma_R^2)- \text{K}_{\text{Var}}- K(T),
\end{equation}
where 
\begin{equation}
\label{kt}
K(T)= \frac{1}{T} \int_0^T \exp(e^{-\alpha s}(X_0-m)+m+\frac{1}{2}v(s)) \,ds.
\end{equation}
In the above expressions $M(s,t)$, $N(s,t)$, and $v(s)$ are given by \eqref{Mttt}, \eqref{Nttt}, and \eqref{vt}, respectively.
\end{theorem}

\begin{proof}
We denote $\tilde{P}(t, \sigma_t^2)= e^{r(T-t)}P(t, \sigma_t^2)$. Then, 
$\tilde{P}(t, \sigma_t^2)= E^{\mathbb{Q}}\left[\left(\sigma_R^2- K_{\text{Var}}\right)\big| \mathcal{F}_t \right]$, is clearly a martingale by construction. We observe, by \eqref{sigmat}, 
\begin{align*}
\tilde{P}(t, \sigma_t^2) &= E^{\mathbb{Q}}\left[\left( \frac{1}{T} \int_0^T ( m_s + e^{X_s}+ \int_0^s e^{-\lambda (s-u)}\,dZ_u)\, ds - K_{\text{Var}}\right)\big| \mathcal{F}_t \right] \\
& = \tilde{P}_1(t) + \tilde{P}_2(t, Z_t),
\end{align*}
where 
\begin{align}
\label{p1}
\tilde{P}_1(t)=  E^{\mathbb{Q}}\left[ \frac{1}{T} \int_0^T e^{X_s} \,ds \big| \mathcal{F}_t \right]& = \frac{1}{T} \int_0^T E^{\mathbb{Q}}(e^{X_s} | \mathcal{F}_t)\, ds \nonumber\\
&= \frac{1}{T} \left(\int_t^T E^{\mathbb{Q}}(e^{X_s} | \mathcal{F}_t)\, ds + \int_0^t E^{\mathbb{Q}}(e^{X_s} | \mathcal{F}_t)\, ds \right) \nonumber \\
&= \frac{1}{T} \left(\int_t^T E^{\mathbb{Q}}(e^{X_s} | \mathcal{F}_t)\, ds + \int_0^t E^{\mathbb{Q}}(e^{X_s})\, ds \right) \nonumber \\
&= \frac{1}{T} \left(\int_t^T e^{M(s, t)+\frac{1}{2}V(s, t)} \, ds + \int_0^t e^{e^{-\alpha s}(X_0-m)+m+\frac{1}{2}v(s)}\, ds \right), 
\end{align}
where $M(s,t)$, $N(s,t)$, and $v(s)$ are given by \eqref{Mttt}, \eqref{Nttt}, and \eqref{vt}, respectively; and
\begin{equation}
\label{p2}
\tilde{P}_2(t,  Z_t)= E^{\mathbb{Q}}\left[\left( \frac{1}{T} \int_0^T ( m_s  + \int_0^s e^{-\lambda (s-u)}\,dZ_u)\, ds - K_{\text{Var}}\right)\big| \mathcal{F}_t \right].
\end{equation}
We denote $P_1(t)= e^{-r(T-t)}\tilde{P}_1(t)$, and $P_2(t,  Z_t)= e^{-r(T-t)}\tilde{P}_2(t,  Z_t)$. We observe that, $P_1(T)= \tilde{P}_1(T)=  \frac{1}{T} \int_0^T E^{\mathbb{Q}}(e^{X_s} | \mathcal{F}_T)\, ds =\frac{1}{T} \int_0^T E^{\mathbb{Q}}(e^{X_s} )\, ds= K(T)$, where $K(T)$ is given by \eqref{kt}. Thus $P_2(T,  Z_T)= \tilde{P}_2(T,  Z_T)= P(T, \sigma^2_T)- K(T)= E^{\mathbb{Q}}(\sigma_R^2)- \text{K}_{\text{Var}}- K(T)$. We note that, from the construction, both $\tilde{P}_1(t)$ and $\tilde{P}_2(t,  Z_t)$ are martingales. We observe
\begin{align*}
& d\tilde{P}_2(t,  Z_t) = d( e^{r(T-t)}P_2(t,  Z_t)) \\
&=  e^{r(T-t)}\left[\left(-rP_2+ \frac{\partial P_2}{\partial t} + \int_0^{\infty} \left(P_2(t, Z_{t-}+y)- P_2(t, Z_{t-})- y \frac{\partial P_2}{\partial z} 1_{|y| \leq 1}\right) \nu(dy)\right)\,dt     \right. \\
&+ \left. \int_0^{\infty} \left(P_2(t, Z_{t-}+y)- P_2(t, Z_{t-})\right)\tilde{J}_Z(dt, dy) \right],
\end{align*}
where $\frac{\partial P_2}{\partial z} $ represents the partial derivative of $P_2$ with respect to its second component, and $\tilde{J}_Z$ is the compensated Poisson random measure for $Z$. 
As $\tilde{P}_2$ and the second integral on the right hand side are martingales, therefore $P_2$ satisfies \eqref{p2eq}.
\end{proof}

\begin{remark}
\label{rem1}
The process $P_2(t,  Z_t)$ in Theorem \ref{thm1} may be thought of as a variance swap price for swaps that are dependent only on the ``jump" process.
\end{remark}

We now propose a hedging procedure for the stock price $S_t$ that satisfies  \emph{fractional BN-S stochastic volatility model} given by \eqref{se500}, \eqref{se501}, \eqref{se502}, and \eqref{se503}. For this we introduce a ``fractional stable process" $Y_t$ given by (with respect to $\mathbb{Q}$) the exponential of fractional Ornstein-Uhlenbeck process. 
\begin{equation}
\label{Yt}
Y_t= Y_0 e^{X_t},
\end{equation}
where $X_t$ is given by \eqref{se503}. We assume $Y_0=1$. We denote the integrated process by $Y_R(t)= \frac{1}{T} \int_0^T Y_s\,ds$. Then the arbitrage-free pricing of such process is given by 
$P_R(t)= e^{-r(T-t)}\tilde{P}_1(t)= P_1(t)$, where $P_1(t)$ is given by \eqref{p1}. Consequently, the price of $Y_R(t)$ may be considered as the process $P_1(t)$. For the follwoing theorem, we denote the L\'evy measure of $Z$ by $\nu_Z(\cdot)$. 

Consider a trading strategy, where the variance swap price $P_2(t,Z_t)$ (as described in Remark \ref{rem1}) and the price process $P_1(t)$ of $Y_R(t)$ are taken as a long and short positions, respectively. Then the following result provides an optimal hedging strategy in terms of the minimization of the quadratic hedging error. 
\begin{theorem}
\label{biggy1}
The risk-minimizing quadratic hedge amounts to holding a position of the underlying $S$ equal to $\phi_t= \Delta(t, S_t, Y_t)$, where
\begin{equation}
\label{optphit}
 \Delta(t, S_t, Y_t)= \frac{\frac{1}{S_t}\int_{\mathbb{R}_{+}} \left( \hat{P}_2(t, Z_{t-}+x)-\hat{P}_2(t, Z_{t-}) \right) (e^{\rho x}-1) \nu_Z(dx)}{\sigma_t^2 +\int_{\mathbb{R}_{+}} (e^{\rho x}-1)^2 \nu_Z(dx) }.
\end{equation}
\end{theorem}

\begin{proof}
Under $\mathbb{Q}$, the discounted commodity price $\hat{S}_t= e^{-rt} S_t$ is a martingale. We consider a self financing strategy $(\phi_t^0, \phi_t)$ with $\phi \in L^2(\hat{S})$, where 
$$\hat{\mathcal{S}}= \{ \phi \text{  predictable  and  } E^{\mathbb{Q}}| \int_0^T \phi_t\, d\hat{S}_t|^2 < \infty \}.$$ The discounted value of the portfolio ($\hat{\Pi}$) is then a martingale with terminal value given by
\begin{align}
\label{alig1}
\hat{\Pi}_T(\phi) & = \int_0^T \phi_t\, d\hat{S}_t=  \int_0^T \phi_t \hat{S}_t \left(  \sigma_t\,dW_t + \int_{\mathbb{R}_{+}}(e^{\rho x}-1)\tilde{J}_{Z}(dt, dx)\right) \nonumber \\
& =  \int_0^T \phi_t \hat{S}_t \sigma_t\,dW_t +  \int_0^{T} \phi_t \hat{S}_t  \int_{\mathbb{R}_{+}}(e^{\rho x}-1)\tilde{J}_{Z}(dt, dx),
\end{align}
where $\tilde{J}_{Z}(\cdot, \cdot)$ is the compensated Poisson measure associated with $Z$.  Next, we consider a variance swap written on $S_t$, and denote $\hat{P}(t, \sigma_t^2)= e^{-rt}P(t, \sigma_t^2)$, $\hat{\tilde{P}}(t, \sigma_t^2)= e^{-rt}\tilde{P}(t, \sigma_t^2)$, and $\Pi_0= e^{-rT} \tilde{P}(0,\sigma_0^2)=P(0,\sigma_0^2)$, then, by It\^o formula and the proof of the Theorem \ref{thm1}, we obtain 
\begin{align}
\label{alig3}
e^{-rT}\tilde{P}(t,\sigma_t^2)-\Pi_0 = \hat{P}_1(t)+ \int_0^{t} \int_{\mathbb{R}_{+}}\left(\hat{P}_2(s, Z_{s-}+x)-\hat{P}_2(s, Z_{s-})\right)  \tilde{J}_{Z}(ds, dx).
\end{align}

Next, we go short in $Y_R(t)$. As the price of $Y_R(t)$ is shown to satisfy $P_1(t)$, we obtain $\hat{Y}_R(T) -\hat{Y}_R(0)= \hat{P}_1(T)$.  We denote the hedging error by $\epsilon(\phi, \Pi_0)= \hat{\Pi}_T(\phi)+ \Pi_{0}-  \hat{\tilde{P}}(T, \sigma_T^2 )+ \hat{P}_1(T)$. We note that $\tilde{P}(T,\sigma_T^2)= P(T,\sigma_T^2)$, and thus we have
\begin{equation}
\label{errorterm}
\epsilon(\phi, \Pi_0)= \hat{\Pi}_T(\phi)+ \Pi_{0}- \hat{P}(T, \sigma_T^2 )+ \hat{P}_1(T).
\end{equation} Considering expressions in \eqref{alig3} at $t=T$, and subtracting that from \eqref{alig1} we obtain
\begin{align*}
 \epsilon(\phi, \Pi_0) & =  \int_0^T \phi_t \hat{S}_t \sigma_t \,dW_t\\
& + \int_0^{T} \int_{\mathbb{R}_{+}}\left[\phi_t \hat{S}_t  (e^{\rho x}-1)-\left( \hat{P}_2(s, Z_{s-}+x)-\hat{P}_2(s, Z_{s-}) \right) \right]\tilde{J}_{Z}(dt, dx).
\end{align*}
 Using isometry formula we obtain the variance of $\epsilon(\phi, \Pi_0) $ as 
\begin{align*}
& E^{\mathbb{Q}}[\epsilon(\phi, \Pi_0)]^2 =  E^{\mathbb{Q}} \left[ \int_0^T \phi_t^2 \hat{S}_t^2  \sigma_t^2 \,dt\right] \\
& + E^{\mathbb{Q}} \left[\int_0^{T} \int_{\mathbb{R}_{+}}\left[\phi_t \hat{S}_t  (e^{\rho x}-1)-\left( \hat{P}_2(s, Z_{s-}+x)-\hat{P}_2(s, Z_{s-}) \right) \right]^2  \nu_Z(dx)\,dt \right].
\end{align*}
The optimal (risk-minimizing) hedge is obtained by minimizing this expression with respect to $\phi_t$. Differentiating the quadratic expression we obtain the first order condition
\begin{align}
\label{buf}
2 \phi_t \hat{S}_t^2  \sigma_t^2+ 2 \int_{\mathbb{R}_{+}}\left[\phi_t \hat{S}_t  (e^{\rho x}-1)-\left( \hat{P}_2(t, Z_{t-}+x)-\hat{P}_2(t, Z_{t-}) \right) \right]\hat{S}_t  (e^{\rho x}-1)  \nu_Z(dx)=0.
\end{align}
Also, in this case, the second-order condition is positive, which confirms the minimization. Solution of \eqref{buf} is given by \eqref{optphit}.
\end{proof}

\section{Numerical Analysis}
\label{sec5}

In this section, we use the VIX data for numerical analysis. The VIX, which is officially known as the Chicago Board Options Exchange (CBOE) Volatility Index, is considered to be an estimator of fear and greed in the stock market (see \cite{vixt}). More precisely, VIX measures the implied volatility in S\&P 500 options. Through the use of a wide variety of option prices, the index gives an estimation of thirty-day implied volatility as priced by the S\&P 500 index options market. It also is a prominent candidate for implementing the \emph{fractional BN-S stochastic volatility model} with $\rho<0$ in \eqref{se501}, as the VIX has an inverse relationship to the performance of the S\&P 500 - a fall of the latter usually corresponds to a rise of the former.

We take an input of the 2008 and 2020 VIX data, and compare how our formulas for the conditional expectation of realized variance compares to the actual realized variance. Both of these years experienced large jumps in volatility, so we are interested in how the model deals with the data across a year. We also include how other models in the literature perform in estimating these quantities to see if our model has improved from previous models. We assume for our fractional BN-S model that $\sigma_0^2=e^{X_0}$. Again, the intuition here is that the process has not started shortly after a jump. We collect the three models we will compare below.
\begin{align*}
    \text{Heston:} &\qquad d\sigma_t^2=\alpha(e^{m}-\sigma_t^2)dt+\gamma\sigma_t dW_t,\\
    \text{BN-S:} &\qquad d\sigma_t^2=-\lambda\sigma_t^2 dt + dZ_{\lambda t},\\
    \text{Fractional BN-S:} &\qquad d\sigma_t^2=\lambda(e^{X_t}-\sigma_t^2)dt+e^{X_t}dX_t+dZ_{t},\quad dX_t=\alpha(m-X_t)dt+\nu dW_t^H,
\end{align*}
and their corresponding expected realized variances
\begin{align*}
    \text{Heston:} &\qquad E^{\mathbb{Q}}[\sigma_R^2]=\frac{1-e^{-\alpha T}}{\alpha T}(\sigma_0^2-e^{m})+e^{m},\\
    \text{BN-S:} &\qquad E^{\mathbb{Q}}[\sigma_R^2]=\frac{1}{T}\left(\lambda^{-1}(1-e^{-\lambda T})\sigma_0^2+\kappa_1(T-\lambda^{-1}(1-e^{-\lambda T}))\right)+\lambda\kappa_2,\\
    \text{Fractional BN-S:} &\qquad E^{\mathbb{Q}}[\sigma_R^2]=\frac{\kappa_1}{\lambda T}\left(T-\lambda^{-1}(1-e^{-\lambda T}\right)+\frac{1}{T}\int_{0}^{T}e^{e^{-\alpha t}(X_0-m)+m+\frac{1}{2}v(t)}dt.
\end{align*}
We assume that the continuous mean reversion rates of the Heston and fractional Ornstein-Uhlenbeck processes are the same, $\alpha$. We also assume that the reversion to mean after jumps of the BN-S and fractional BN-S are the same, $\lambda$. This leads to more fair comparisons between the two models.

Let us now describe the process we used to generate the following curves. There is a computational difficulty here, in that as can be seen with \eqref{Mttt} and \eqref{Nttt}, there are improper integrals involved in the computation of forward curves. These integrals admit no nice closed form, and thus it became computationally impractical to use the full history in simulations. Instead, we only use the non-conditional expectations derived earlier in the paper which we can use the initial condition as the real-world input.

We define a moving window of $T$ days, and we plot the actual realized variance across these $T$ days. Then we also plot what the expected realized variance is for those $T$ days if we assumed the initial condition $\sigma_0^2=\text{Average of previous } T \text{ days}$. This allows the model to still make some use of the history while also allowing the computations to be done in a reasonable amount of time. Finally, we keep track of several different errors between the models and the real-world data to compare the fractional BN-S with other models.

Let us quickly catalog the parameters we used to generate the following graphs: $\alpha = 1/30,\, m=\log(400),\,\nu=1/3,\,H=3/4,\,\lambda=1/20,\,\kappa_1=30,\,\kappa_2=25.$

Figures 1 and 2 show the massive spike in the variance in 2008. The dates are January 1st, 2008 to December 31st, 2008. Here we can see that the fractional model starts to perform significantly better around these large spikes. Similarly, in Figures 3 and 4, we have the jump in volatility from early 2020. The dates are January 1st, 2020 to December 31st, 2020. Here we also see how the fractional model performs better. The long-term memory can be seen well in figure 4 as the fractional model stays higher after a spike from its long memory.

\begin{figure}[H]
\centering
\includegraphics[scale=0.8]{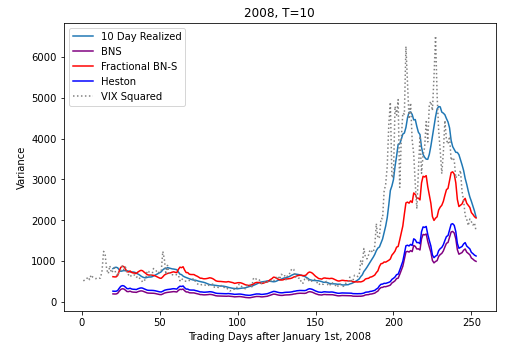}
\caption{Estimated 10-day realized variance in 2008.}
\end{figure}

\begin{figure}[H]
\centering
\includegraphics[scale=0.8]{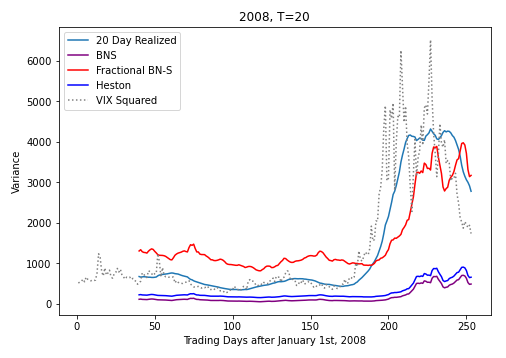}
\caption{Estimated 20-day realized variance in 2008.}
\end{figure}

\begin{figure}[H]
\centering
\includegraphics[scale=0.8]{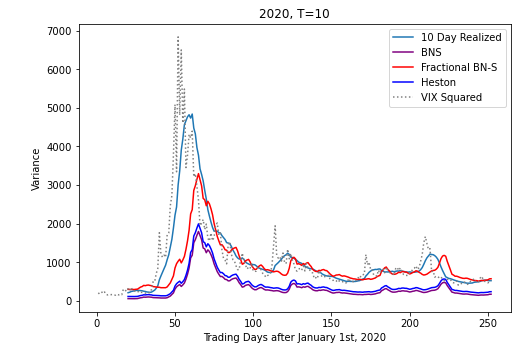}
\caption{Estimated 10-day realized variance in 2020.}
\end{figure}

\begin{figure}[H]
\centering
\includegraphics[scale=0.8]{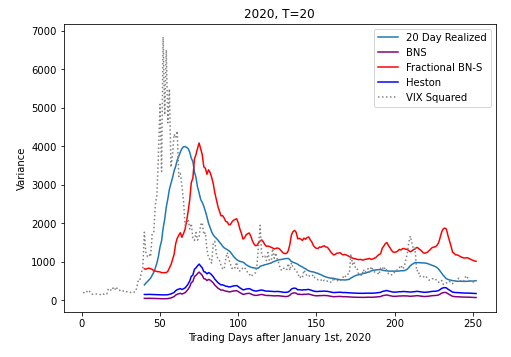}
\caption{Estimated 20-day realized variance in 2020.}
\end{figure}

For the goodness of fit of the above graphs, we use the absolute percentage error (APE), the average absolute error (AAE), the average relative percentage error (ARPE) and the root-mean-square error (RMSE) given by the following formulas.
\begin{equation*}
\text{APE} = \frac{1}{\text{mean value}} \sum_{\text{data points}} \frac{|\text{data value} - \text{model value}|}{\text{data points}},
\end{equation*}
\begin{equation*}
\text{AAE} =  \sum_{\text{data points}} \frac{|\text{data value}- \text{model value}|}{\text{data points}},
\end{equation*}
\begin{equation*}
\text{ARPE} =  \frac{1}{\text{data points}} \sum_{\text{data points}} \frac{|\text{data value}- \text{model value}|}{\text{data points}},
\end{equation*}
\begin{equation*}
\text{RMSE} =  \sqrt { \sum_{\text{data points}} \frac{(\text{data value}- \text{model value})^2}{\text{data points}}}.
\end{equation*}

Tables 1, 2, 3, and 4, correspond to the goodness of fit estimates from Figures 1, 2, 3, and 4, respectively.
\begin{table}[H] 
\centering 
\begin{tabular}{|c|| c| c| c|c| }       
 \hline
  Model & APE & AAE & ARPE & RMSE\\
 \hline
 Heston & 2.17 & 866.89 & 3.70 & 1312.38\\
 \hline
 BN-S & 2.39 & 952.86 & 4.07 & 1392.71\\
 \hline
 Fractional BN-S & 1.12 & 448.84 & 1.91 & 824.41\\
 \hline
\end{tabular}
\caption{Errors for estimated 10-day realized variance in 2008.}
\end{table}

\begin{table}[H] {
\centering 
\begin{tabular}{|c|| c| c| c|c| }       
 \hline
 Model & APE & AAE & ARPE & RMSE\\
 \hline
 Heston & 2.76 & 1101.33 & 5.14 & 1660.11\\
 \hline
 BN-S & 3.06 & 1221.49 & 5.70 & 1758.76\\
 \hline
 Fractional BN-S & 1.63 & 651.74 & 3.04 & 743.99\\
 \hline
\end{tabular}
\caption{Errors for estimated 20-day realized variance in 2008.}}
\end{table}

\begin{table}[H] {
\centering 
\begin{tabular}{|c|| c| c| c|c| }       
 \hline
 Model & APE & AAE & ARPE & RMSE\\
 \hline
 Heston & 1.64 & 659.78 & 2.83 & 984.54\\
 \hline
 BN-S & 1.84 & 739.25 & 3.17 & 1049.34\\
 \hline
 Fractional BN-S & 0.74 & 296.23 & 1.27 & 665.81\\
 \hline
\end{tabular}
\caption{Errors for estimated 10-day realized variance in 2020.}}
\end{table}

\begin{table}[H] {
\centering 
\begin{tabular}{|c|| c| c| c|c| }       
 \hline
 Model & APE & AAE & ARPE & RMSE\\
 \hline
 Heston & 2.20 & 881.49 & 4.13 & 1190.60\\
 \hline
 BN-S & 2.49 & 999.63 & 4.69 & 1285.62\\
 \hline
 Fractional BN-S & 1.82 & 729.25 & 3.42 & 858.36\\
 \hline
\end{tabular}
\caption{Errors for estimated 20-day realized variance in 2020.}}
\end{table}

As we can see the \emph{fractional BN-S model} does best in comparison to other models when very large deviations from the mean have happened. This is because the model is combining the addition of upward jumps with long-term memory with positive autocorrelation. Hence during these large spikes in the variance, the predicted realized variance will be larger than other models and as we can see be more accurate. It should be noted that this numerical analysis is simplified, and with more computing power we should find even better accuracy for the fractional BN-S model. One of the main points of the model is the long-term memory, and this analysis could only take into account the initial condition rather than the whole history. It seems likely that if this whole history was taken into account that even more interesting behavior would be observed.

\section{Conclusion}
\label{sec6}

In this paper, we introduce and analyze the \emph{fractional BN-S} model. This model builds upon two desirable properties of the long-term variance process in the empirical data: long-term memory and jumps. It takes the long-term memory and positive autocorrelation properties of fractional Brownian motion with $H>1/2$ and the jump properties of the BN-S model. In the process of finding arbitrage-free prices for variance and volatility swaps, we derive some new expressions for the distributions of integrals of continuous Gaussian processes and use them to get an explicit analytic expression for the fair prices of these swaps. We then proceed to obtain analytic expressions for conditional distributions of the realized variance. These expressions turned out to be highly singular, hence simplified methods are used to test the \emph{fractional BN-S model} against real-world data. We find that the \emph{fractional BN-S model} appears to outperform the Heston model and classical BN-S model. Future research with more powerful computational facilities or more optimized numerical methods may lead to more interesting observations from these real-world tests.

\end{document}